\documentclass[11pt]{article}

\usepackage[utf8]{inputenc} 
\usepackage{latexsym}
\usepackage{amssymb}
\usepackage{mathrsfs}
\usepackage[table,xcdraw]{xcolor}
\usepackage{amsmath}
\usepackage{amsthm}

\theoremstyle{definition}
\newtheorem{defi}{Definition}[section]
\newtheorem{rmk}{Remark}[section]
\newtheorem{prop}{Proposition}[section]

\newtheorem{thm}{Theorem}[section]
\newtheorem{eg}{Example}[section]
\newtheorem{ax}{Axiom}

\newtheorem{lemma}{Lemma}[section]

\newtheoremstyle{mystyle}
{3pt}
{3pt}
{\upshape}
{}
{}
{:}
{.5em}
{}
\theoremstyle{mystyle}

\usepackage{pythonhighlight}
\usepackage{geometry} 
\usepackage{graphicx} 
\usepackage{booktabs} 
\usepackage{array} 
\usepackage{paralist} 
\usepackage{verbatim} 
\usepackage{subfig} 

\usepackage{fancyhdr} 
\usepackage{sectsty}
\allsectionsfont{\sffamily\mdseries\upshape} 

\usepackage[nottoc,notlof,notlot]{tocbibind} 
\usepackage[titles,subfigure]{tocloft} 

\usepackage[colorlinks,linkcolor=blue]{hyperref}
\usepackage[all]{xy}
\usepackage{cancel}
\newenvironment{keywords}
{\par\noindent\textbf{Keywords:}}
{\par}

\title{Spin Statistics and Field Equations for any Spin}
\author{Feng Zixuan}

\begin{document}
\maketitle
\tableofcontents

\begin{abstract}
    In this article we prove spin statistics theorem for arbitrary massive $(A,B)$ field in a representation theoretic manner. General Gamma matrices are introduced, and explicit forms for low spin are calculated. Spin sums and twisted spin sums are introduced to prove spin statistics and derive field equations respectively, and their relations are discussed. Klein Gordan equation is just the condition on the 4-momentum. Dirac equation (or massive Weyl equations) and Proca equation are shown to be examples of our general field equation. 
\end{abstract}

\clearpage

\section{Introduction}

In this article we prove the spin statistics theorem in a representation theoretic manner for arbitrary field spin $(A,B)$ for massive field. Also we derive general field equation for massive field, and shows that the Dirac equation and Proca equation are examples for it.

The logic chain of this article follows Weinberg's path (\cite{weinberg1964feynman}) of deriving quantum field theory. We first have a unitary representation of Poincare group on the state space (Fock space), then we define creation and annilation operators explicitly, and then use them to construct free quantum fields. In this approach the field equation does not serve as an equation that characterizes the field uniquely, but is a property that the field possesses. So field equation is not of great importance from certain point of view. 

\begin{equation}
    \text{Fock space representation} \longrightarrow
    \text{Creation and annilation} \longrightarrow
    \text{Free quantum field} \longrightarrow
    \text{Field equation}
\end{equation}

First we recall the construction for free quantum field. 
We define annihilation field $\psi_{l}^{+}(x)$ and creation field $\psi_l^{-}(x)$:
\begin{equation}
    \begin{aligned}
    & \psi_{l}^{+}(x)=\sum_{\sigma n} \int d^3 p u_{l}(x ; \mathbf{p}, \sigma, n) a(\mathbf{p}, \sigma, n), \\
    & \psi_{l}^{-}(x)=\sum_{\sigma n} \int d^3 p v_{l}(x ; \mathbf{p}, \sigma, n) a^{\dagger}(\mathbf{p}, \sigma, n)
    \end{aligned}
\end{equation}
We hope it to satisfy the 
Lorentz transformation of field:
\footnote{Our convention is different from that of \cite{weinberg2002quantum}.}
\begin{ax}(Lorentz transformation of field)

    We have a finite dimensional representation $D$ of the universal cover $SL(2,\mathbb{C})$ of the Lorentz group such that the 'Lorentz transformation of the field' are satisfied:
    \begin{equation}
    \begin{aligned}
        U_0(\Lambda, a)^{-1} \psi_{l}^{+}(\lambda(\Lambda) x+a) U_0(\Lambda, a)
        = D_{l \bar{l}}(\Lambda) \psi_{\bar{l}}^{+}(x)\\
        U_0(\Lambda, a)^{-1} \psi_{l}^{-}(\lambda(\Lambda) x+a) U_0(\Lambda, a)
        = D_{l \bar{l}}(\Lambda) \psi_{\bar{l}}^{-}(x)\\
        \forall \Lambda\in SL(2,\mathbb{C})
    \end{aligned}
    \end{equation}
\end{ax}

In order for Axiom 1 to satisfy, it is equivalent that
\begin{equation}\label{B}
    \begin{aligned}
        u(x; \mathbf{q})=(2\pi)^{-3 / 2} e^{ip\cdot x} u(\mathbf{q})
        \\
        v(x; \mathbf{q})=(2\pi)^{-3 / 2} e^{-ip\cdot x} v(\mathbf{q})
    \end{aligned}
\end{equation}
and:
\begin{equation}\label{B}
    \begin{aligned}
        u(\mathbf{q})=(m / q^0)^{1 / 2} D(L(q)) u(\mathbf{0})
        \\
        v(\mathbf{q})=(m / q^0)^{1 / 2} D(L(q)) v(\mathbf{0})
    \end{aligned}
\end{equation}
and:
\begin{equation}
    \begin{aligned}\label{C}
        u(\mathbf{0}) D^{(j_n)}(W)
        =D(W) u(\mathbf{0})
        \\
        v(\mathbf{0}) D^{(j_n)^*}(W)
        =D(W) v(\mathbf{0})
    \end{aligned}
\end{equation}
$\forall W\in \mathcal{W}$, 
where we have used the convention: $u(q) = u(\mathbf{q})$ denotes the matrix $\{u_l(\mathbf{q},\sigma)\}_{l|\sigma}$, and the multiplication is the matrix multiplication. 

Considering a irreducible field representation of spin $(A,B)$, we write

\begin{equation}
    \begin{aligned}
    \psi^{AB}_{a b}(x)= & (2 \pi)^{-3 / 2} \sum_\sigma \int d^3 p
    [e^{i p \cdot x} \kappa^{AB}  u^{AB}_{a b}(\mathbf{p}, \sigma) a(\mathbf{p}, \sigma)
    + e^{-i p \cdot x} \lambda^{AB} v^{AB}_{a b}(\mathbf{p}, \sigma) a^{c \dagger}(\mathbf{p}, \sigma)]
    \end{aligned}
\end{equation}
with $\kappa$ and $\lambda$ arbitrary constants waiting for determination.

The next property we want the quantum field to possess is the microscopic causality:

\begin{ax}(Microscopic causality)

    We should have
    \begin{equation}
        [\psi^{AB}_{a b}(x), \psi^{CD\dagger}_{cd}(y)]_{\mp} = 0, \quad (x-y)^2>0
    \end{equation}
\end{ax}

This requirement leads to the so called 'spin statistics theorem', which is discussed in Section 3. 
The central object here is the so called 'spin sum':
\begin{equation}
    (2 p^0)^{-1} \pi^{AB,CD}(\mathbf{p})
    \equiv  u^{AB}(\mathbf{p}) u^{CD\dagger}(\mathbf{p})
    = v^{AB}(\mathbf{p}) v^{CD\dagger}(\mathbf{p})
\end{equation}
and we can show that it may be expressed as a covariant polynomial of $p^\mu$ using generalized gamma matrices, which we will discuss in Section 2. 
As a result we can fix a normalization of $\kappa,\lambda$:
\begin{equation}
    \begin{aligned}
    \psi^{AB}_{a b}(x)= & (2 \pi)^{-3 / 2} \sum_\sigma \int d^3 p
    [e^{i p \cdot x}  u^{AB}_{a b}(\mathbf{p}, \sigma) a(\mathbf{p}, \sigma)
    + e^{-i p \cdot x} (-)^{2B} v^{AB}_{a b}(\mathbf{p}, \sigma) a^{c \dagger}(\mathbf{p}, \sigma)]
    \end{aligned}
\end{equation}
where the factor $(-)^{2B}$ is essential for both spin statistics and the existence for a field equation where we want to combine the two equations for positive frequency part and the negative frequency part, as discussed in Section 4. In writing down a field equation, we use a modified version of spin sum, called 'twisted spin sum' in this article.



\section{Gamma matrices}

\subsection{General spin}

We can prove that for any half integers $A,B,C,D,K$ satisfying
\begin{equation}\label{ABCDK}
    \max\{|A-D|,|B-C|\} \le K \le \min\{A+D,B+C\}
\end{equation}
with $2K,2A+2D,2B+2C$ having the same parity,
\footnote{This parity condition will be implicit when we write down the condition \eqref{ABCDK} afterwards. }
there exists a set of non-zero $(2A+1)(2B+1)\times (2C+1)(2D+1)$ matrices
\begin{equation}
    ^{ABCD}_{\quad K}T^{\mu_1\mu_2 \cdots \mu_{2K}}, 
     \quad \mu_1,\mu_2, \cdots ,\mu_{2K} = 0,1,2,3
\end{equation}
with the properties:
\begin{itemize}
    \item 
    $T$ is symmetric in $\mu$'s.
    \item 
    $T$ is traceless in $\mu$'s, i.e.,
    \begin{equation}
        g_{\mu_1\mu_2}T^{\mu_1\mu_2 \cdots \mu_{2K}}=0
    \end{equation}
    \item
    $T$ is a tensor in $\mu$'s, i.e.,
    \begin{equation}\label{Transformation of T matrices}
        D^{AB}(\Lambda) T^{\mu_1\mu_2 \cdots \mu_{2K}} D^{CD}(\Lambda)^\dagger
        = \Lambda_{\nu_1}{}^{\mu_1} \Lambda_{\nu_2}{}^{\mu_2} \cdots \Lambda_{\nu_{2K}}{}^{\mu_{2K}}
        T^{\nu_1\nu_2 \cdots \nu_{2K}}
    \end{equation}
    where the ordinary matrix multiplication is understood on the LHS. 
\end{itemize}

Consider the vector space $V$ consisting of all complex $(2A+1)(2B+1)\times(2C+1)(2D+1)$-matrices. They furnish a representation of $SL(2,\mathbb{C})$ by
\begin{equation}
    M \mapsto D^{AB}(\Lambda) M D^{CD\dagger}(\Lambda)
\end{equation}
One can see that this is isomorphic to the $(A,B)\otimes(D,C)$ representation, which is because the complex conjugate of $D^{CD}$ is isomorphic (but not directly equal) to $D^{DC}$. 
And we have
\footnote{From now on we use the convention.}
\begin{equation}
\begin{aligned}
    (A,B)\otimes(D,C) = (A\otimes D , B\otimes C ) &= (\bigoplus_{|A-D|\le K_1\le A+D} K_1 , \bigoplus_{|B-C|\le K_2\le B+C} K_2 )\\
    &= \bigoplus_{|A-D|\le K_1\le A+D, |B-C|\le K_2\le B+C} (K_1 , K_2 )
\end{aligned}
\end{equation}
So it contains a $(K,K)$-subrep. By the fact
\footnote{See \cite{weinberg1964feynman}}
that $(K,K)$-representation of $SL(2,\mathbb{C})$ consists of all symmetric traceless tensors of rank $2K$, so we define $T$'s to be the standard basis in this description and we are done.

\subsection{Examples}

Pauli matrices $\sigma^\mu$ are the $T$ matrices for $(\frac{1}{2},0,\frac{1}{2},0)$, i.e. $\sigma^\mu = ^{1/2,0,1/2,0}_{\quad 1/2} T^\mu$, where
\begin{equation}
    \sigma^0 = I = \begin{pmatrix}
    1 & 0\\
    0 & 1
    \end{pmatrix}
    \quad
    \sigma^1 = X = \begin{pmatrix}
    0 & 1\\
    1 & 0
    \end{pmatrix}
    \quad
    \sigma^2 = Y = \begin{pmatrix}
    0 & -i\\
    i & 0
    \end{pmatrix}
    \quad
    \sigma^3 = Z = \begin{pmatrix}
    1 & 0\\
    0 & -1
    \end{pmatrix}
\end{equation}

Indeed we see in the decomposition $(\frac{1}{2},0)\otimes (0,\frac{1}{2})=(\frac{1}{2},\frac{1}{2})$ there is no other subreps, so we only need to show they obey the property
\begin{equation}
    D^{1/2,0}(\Lambda) \sigma^\mu D^{1/2,0\dagger}(\Lambda) = \Lambda_\nu{}^\mu \sigma^\nu
\end{equation}
or its Lie algebra level (where $J^{\mu\nu}$ denotes the image of Lie algebra basis under representation):
\begin{equation}
    i [J^{\mu\nu}, \sigma^\rho ] = -g^{\rho\mu} \sigma^\nu + g^{\rho\nu}\sigma^\mu
\end{equation}
This is by the definition of $(\frac{1}{2},0)$ representation:
\begin{equation}
    \begin{aligned}
        J_1 &\mapsto \frac{1}{2}X\\
        J_2 &\mapsto \frac{1}{2}Y\\
        J_3 &\mapsto \frac{1}{2}Z\\
        K_1 &\mapsto -\frac{i}{2}X\\
        K_2 &\mapsto -\frac{i}{2}Y\\
        K_3 &\mapsto -\frac{i}{2}Z
    \end{aligned}
\end{equation}

Similarly we can show that $\bar{\sigma}^\mu \equiv (I, -\mathbf{\sigma})$ are the $T$ matrices for $(0,\frac{1}{2},0,\frac{1}{2})$, i.e. $\bar{\sigma}^\mu = ^{0,1/2,0,1/2}_{\quad 1/2} T^\mu$.


\section{Spin statistics for massive particles}

The results in this subsection apply only when the field representation is irreducible and the particle is massive.

Let's compute $[\psi^{AB}_{a b}(x), \psi^{CD\dagger}_{c d}(y)]_{\mp}$. Firstly we have the equality 
\begin{equation}
    \begin{aligned}
    [\psi^{AB}_{a b}(x), \psi^{CD\dagger}_{c d}(y)]_{\mp}
    = & (2 \pi)^{-3 / 2} \int d^3 p(2 p^0)^{-1} 
    \pi^{AB,CD}_{a b, c d}(\mathbf{p}) \\
    & \times[\kappa^{AB} \kappa^{CD*} e^{i p \cdot(x-y)} \mp \lambda^{AB} \lambda^{CD*} e^{-i p \cdot(x-y)}]
    \end{aligned}
\end{equation}

where 

\begin{defi}(Spin sum)

    Define the spin sum as
    \begin{equation}
    (2 p^0)^{-1} \pi^{AB,CD}(\mathbf{p})
    \equiv  u^{AB}(\mathbf{p}) u^{CD\dagger}(\mathbf{p})
    = v^{AB}(\mathbf{p}) v^{CD\dagger}(\mathbf{p})
    \end{equation}
    It is a $(2A+1)(2B+1)\times(2C+1)(2D+1)$-matrix. 
\end{defi}

\begin{rmk}
    We use the convention that $u^{AB}(p)$ is a matrix $(u_{ab}^{AB}(p,\sigma))_{ab,\sigma}$, having $(2A+1)(2B+1)$ rows and $2j+1$ columns. The above is understood as a matrix multiplication. Written explicitly in components, it reads
    \begin{equation}
    (2 p^0)^{-1} \pi^{AB,CD}_{a b, c d}(\mathbf{p})
    \equiv \sum_\sigma u^{AB}_{a b}(\mathbf{p}, \sigma) u^{CD*}_{cd}(\mathbf{p}, \sigma)
    =\sum_\sigma v^{AB}_{a b}(\mathbf{p}, \sigma) v^{CD*}_{cd}(\mathbf{p}, \sigma)
    \end{equation}
    But for convenience, we adapt the more concise notation involving matrix multiplication. 
\end{rmk}

We immediately see its Lorentz transformation property:
\begin{prop}
    \begin{equation}\label{Transformation of Spin Sum}
        \pi^{AB,CD}(\Lambda p)
        =
        D^{AB}(\Lambda) \pi^{AB,CD}(p) D^{CD}(\Lambda)^\dagger
    \end{equation}
\end{prop}
\begin{proof}
\begin{equation}
\begin{aligned}
    D^{AB}(\Lambda) \pi^{AB,CD}(p) D^{CD\dagger}(\Lambda) 
    &= 2p^0 D^{AB}(\Lambda) u^{AB}(p) u^{CD\dagger}(p) D^{CD\dagger}(\Lambda) \\
    &= 2(\Lambda p)^0 u^{AB}(\Lambda p) D^j(W(\Lambda,p)) D^{j\dagger}(W(\Lambda,p)) u^{CD\dagger}(\Lambda p)\\
    &= 2(\Lambda p)^0 u^{AB}(\Lambda p) u^{CD\dagger}(\Lambda p)\\
    &= \pi^{AB,CD}(\Lambda p)          
\end{aligned}
\end{equation}
\end{proof}

The major part of the proof of spin statistics is that we can express the spin sum as a polynomial with a good parity property, and it is this polynomial's parity that determine the boson/fermion of the particle. The parity is used in absorbing the minus sign in $e^{-ip\cdot x}$ of $\psi^{(-)}$.
\footnote{Until not we do not assume the particle to be massive. }

\begin{thm}
\footnote{Weinberg in \cite{weinberg2002quantum} proved the special case when $\mathbf{p} = (0,0,p^3)$. But it does not directly imply the general case.}

    In the massive case there is a 4-variables polynomial $P$ of coefficients $(2A+1)(2B+1)\times(2C+1)(2D+1)$-matrices, such that on the mass shell its value always equal the spin sum
    \begin{equation}
        \pi^{AB,CD}  (\mathbf{p},\sqrt{\mathbf{p}^2+m^2})
        =P_{}(\mathbf{p},\sqrt{\mathbf{p}^2+m^2})
    \end{equation}
    and it is either an odd function or an even function determined by the parity of $2A+2D$:
    \begin{equation}
        P(-\mathbf{p}, -p^0) = (-1)^{2A+2D} P(\mathbf{p},p^0)
    \end{equation}
    for all 4-vectors $p$.
\end{thm}

Let's prove the theorem.

First we can prove that the initial value $\pi^{AB,CD}(\mathbf{0})$ can be expressed as a linear combination of above matrices
\begin{equation}\label{Initial value of Spin Sum}
    \pi^{AB,CD}(\mathbf{0}) = \sum_{K \text{satisfying }\eqref{ABCDK}} 
    \xi_{K}^{ABCD}{} ^{ABCD}_{\quad K}T^{0_1 0_2 \cdots 0_{2K}}
\end{equation}
This is because by \eqref{Transformation of Spin Sum}, $\pi^{AB,CD}(\mathbf{0})$ is invariant under the representation restricted to $SU(2)$, 
so it must lie in the direct sum of $0$-subrepresentations of $SU(2)$ when we decompose 
\begin{equation}
\begin{aligned}
    (A,B)\otimes(D,C)|_{SU(2)}
    &= \bigoplus_{|A-D|\le K_1\le A+D, |B-C|\le K_2\le B+C} (K_1 , K_2 )|_{SU(2)}\\
    &= \bigoplus_{|A-D|\le K_1\le A+D, |B-C|\le K_2\le B+C} 
    \bigoplus_{|K_1-K_2| \le K_3\le K_1+K_2} K_3\\
\end{aligned}
\end{equation}
We see that each $0$-subrep of $SU(2)$ comes from a $(K,K)$-subrep of $SL(2,\mathbb{C})$ where $K$ satisfies \eqref{ABCDK}. Also, by \eqref{Transformation of T matrices}, $T^{00\cdots 0}$ is non-zero and invariant under $SU(2)$, so they form a basis for the $0$-subspace of $SU(2)$.

Finally, combine \eqref{Transformation of Spin Sum}, \eqref{Transformation of T matrices} and \eqref{Initial value of Spin Sum}, we have the equality on mass shell
\begin{equation}
    \pi^{AB,CD}(\mathbf{p}) = \sum_{K \text{satisfying }\eqref{ABCDK}} 
    \xi_{K}^{ABCD}{} ^{ABCD}_{\quad K} T^{\mu_1\mu_2 \cdots \mu_{2K}} m^{-2K} p_{\mu_1}p_{\mu_2} \cdots p_{\mu_{2K}}
\end{equation}
So we define the polynomial as the RHS.

The proof is done. 
\\\\



With this theorem we can finally begin the proof for spin statistics. First we show that we need only verify a simpler case:

\begin{lemma}\footnote{Weinberg in \cite{weinberg2002quantum} mentioned and used this result.}

    For $[\psi^{AB}_{a b}(x), \psi^{CD\dagger}_{c d}(y)]_{\mp}$ vanishes for space-like $x-y$, it is sufficient that it vanishes when $x^0=y^0$. 
\end{lemma}
\begin{proof}
    We can choose a Lorentz transformation $\Lambda$ such that $(\Lambda x)^0=(\Lambda y)^0$. 
    Under Lorentz transformation we have 
    \begin{equation}
        [\psi^{AB}(x), \psi^{CD\dagger}(y)]_{\mp}
        = U_0(\Lambda)^{-1} D^{AB}(\Lambda^{-1})
        [\psi^{AB}(\Lambda x), \psi^{CD\dagger}(\Lambda y)]_{\mp} D^{CD\dagger}(\Lambda^{-1}) U_0(\Lambda)
    \end{equation}
    So it follows. 
\end{proof}

Now we rewrite the above polynomial: substitute all exponentials of $p^0$ of degree larger than 1 by substituting $(p^0)^2=\mathbf{p}^2+m^2$. Then the value on mass shell do not change (but the polynomial changes) and can be written as
 \begin{equation}
    \pi^{AB, CD}(\mathbf{p}) = P(\mathbf{p}) + 2 \sqrt{\mathbf{p}^2+m^2} Q(\mathbf{p})
\end{equation}
where $P$ and $Q$ are 3-variables polynomials in $\mathbf{p}$ alone, with
\begin{equation}
    \begin{aligned}
    & P(-\mathbf{p})=(-)^{2 A+2 D} P(\mathbf{p}) \\
    & Q(-\mathbf{p})=-(-)^{2 A+2 D} Q(\mathbf{p})
    \end{aligned}
\end{equation}
For $x-y$ space-like, by the above lemma we can adopt a Lorentz frame in which $x^0=y^0$, and thus
\begin{equation}
    \begin{aligned}
    [\psi^{AB}(x), \psi^{CD\dagger}(y)]_\mp
    = & [\kappa^{AB} \kappa^{CD*} \mp (-)^{2 A+2 D} \lambda^{AB} \lambda^{CD*}] P(-i \nabla) \Delta_{+}(\mathbf{x}-\mathbf{y}, 0)  \\
    & +[\kappa^{AB} \kappa^{CD} \pm(-)^{2 A+2 D} \lambda^{AB} \lambda^{CD*}] Q(-i \nabla) \delta^3(\mathbf{x}-\mathbf{y})
    \end{aligned}
\end{equation}

In order that this should vanish when $\mathbf{x} \neq \mathbf{y}$, we must have
\begin{equation}
    \kappa^{AB} \kappa^{CD*}=\pm(-1)^{2 A+2 D} \lambda^{AB} \lambda^{CD*}
\end{equation}

For the case where $A=C$ and $B=D$. 
\begin{equation}
    |\kappa^{AB}|^2=\pm(-)^{2 A+2 B}|\lambda^{AB}|^2
\end{equation}

This is possible if and only if
\begin{equation}
    \pm(-1)^{2 A+2 B}=+1
\end{equation}
and
\begin{equation}
    |\kappa^{AB}|^2=|\lambda^{AB}|^2
\end{equation} 

Return to the general case we have
\begin{equation}
    \frac{\kappa^{AB}}{\kappa^{CD}}=(-1)^{2 A-2 C} \frac{\lambda^{AB}}{\lambda^{CD}}
\end{equation}

To summarize what we have proved above:

\begin{thm}(Spin statistics, massive irreducible case)\label{Main-Spin statistics}
\footnote{This result is stated in \cite{weinberg2002quantum}, but there is no rigorous proof.}
\footnote{The factor $(-)^{2B}$ is irrelevant when we just consider one field, but is essential when considering different fields.}

    In order for all fields $\psi^{AB}(x)$ constructed from a massive particle with spin $j$ to satisfy Axiom 1, it is equivalent that whether it is boson or fermion depends on whether $2j$ is even or odd, 
    and $\lambda^{AB} = (-)^{2B} \kappa^{AB} c$, where $c$ is a constant independent of $A,B$.  
\end{thm}

\begin{eg}$\quad$

        Massive scalar field has $(A,B)=(0,0)$ and $j=0$, so it describes bosons. 
    
        Massive vector field has $(A,B)=(\frac{1}{2}, \frac{1}{2})$ and $j=0,1$, so it describes bosons. 
    
        Massive Weyl field has $(\frac{1}{2}, 0)$ or $(0, \frac{1}{2})$ and $j=\frac{1}{2}$, so it  describes fermions. 

 \end{eg}

So we can choose a normalization for each $(A,B)$ field separately, and for $c$ we can adjust the phase in the definition of creation and annilation operators and $c=1$. As the result we have:
\begin{equation}
    \begin{aligned}
    \psi^{AB}_{a b}(x)= & (2 \pi)^{-3 / 2} \sum_\sigma \int d^3 p
    [e^{i p \cdot x}  u^{AB}_{a b}(\mathbf{p}, \sigma) a(\mathbf{p}, \sigma)
    + e^{-i p \cdot x} (-)^{2B} v^{AB}_{a b}(\mathbf{p}, \sigma) a^{c \dagger}(\mathbf{p}, \sigma)]
    \end{aligned}
\end{equation}


\section{Field equations for massive particles}

\subsection{General equation}

For a quantum field $\psi_l^+(x) = (2\pi)^{-3/2} \int d^3p e^{ip\cdot x} u_l(\mathbf{p},\sigma) a(\mathbf{p},\sigma)$, denote by 
\begin{equation}
\psi^{(+)AB}_l(p)= (2\pi)^{-3/2} u^{AB}_l(\mathbf{p},\sigma) a(\mathbf{p},\sigma)
\end{equation}
its Fourier transformation, and the similarly 
\begin{equation}
\psi^{(-)AB}_l(p)= (2\pi)^{-3/2} (-)^{2B} v^{AB}_l(\mathbf{p},\sigma) a^\dagger(\mathbf{p},\sigma)
\end{equation}
Clearly by unitarity of $D^{CD}|_{SU(2)}$, we have $u^{CD\dagger}(0)u^{CD}(0) = I$, so we have
\begin{equation}\label{Field equation momentum}
\begin{aligned}
    \psi^{(+)AB}(p) &= \Pi^{AB,CD}(p) \psi^{(+)CD}(p)\\
    \psi^{(-)AB}(p) &= (-)^{2B+2D}\Pi^{AB,CD}(p) \psi^{(-)CD}(p)
\end{aligned}
\end{equation}

where we have defined:

\begin{defi}(Twisted spin sum)

Define the twisted spin sum
\begin{equation}
\begin{aligned}
    \Pi^{AB,CD}(p) &\equiv D^{AB}(Lp) u^{AB}(\mathbf{0}) u^{CD\dagger}(\mathbf{0}) D^{CD-1}(Lp)\\
    &= D^{AB}(Lp) v^{AB}(\mathbf{0}) v^{CD\dagger}(\mathbf{0}) D^{CD-1}(Lp)
\end{aligned}
\end{equation}
\end{defi}

We have its transformation rule just like for spin sum:

\begin{prop}
    \begin{equation}
        \Pi^{AB,CD}(\Lambda p)
        =
        D^{AB}(\Lambda) \Pi^{AB,CD}(p) D^{CD-1}(\Lambda)
    \end{equation}
\end{prop}
\begin{proof}
    \begin{equation}
    \begin{aligned}
        &D^{AB}(\Lambda) \Pi^{AB,CD}(p) D^{CD-1}(\Lambda)
        = D^{AB}(\Lambda) D^{AB}(Lp) u^{AB}(0) u^{CD\dagger}(0) D^{CD-1}(Lp) D^{CD-1}(\Lambda)\\
        &= D^{AB}(L(\Lambda p)) D^{AB}(W(\Lambda,p)) u^{AB}(0) u^{CD\dagger}(0) D^{CD-1}(W(\Lambda,p)) D^{CD-1}(L(\Lambda p))\\
        &= D^{AB}(L(\Lambda p)) D^{AB}(W(\Lambda,p)) u^{AB}(0) u^{CD\dagger}(0) D^{CD\dagger}(W(\Lambda,p)) D^{CD-1}(L(\Lambda p))\\
        &= D^{AB}(L(\Lambda p)) u^{AB}(0) D^j(W(\Lambda,p)) D^{j\dagger}(W(\Lambda,p)) u^{CD\dagger}(0)  D^{CD-1}(L(\Lambda p))\\
        &= D^{AB}(L(\Lambda p)) u^{AB}(0)u^{CD\dagger}(0)  D^{CD-1}(L(\Lambda p))\\
        &= \Pi^{AB,CD}(\Lambda p)
    \end{aligned}
    \end{equation}
    where we have used the unitarity of $D^{CD}|_{SU(2)}$ in the third line. 
\end{proof}

So we have the similar result for twisted spin sum:

\begin{thm}$\quad$
\label{Twisted}

    In the massive case there is a 4-variables polynomial $P$ of coefficients $(2A+1)(2B+1)\times(2C+1)(2D+1)$-matrices, such that on the mass shell its value always equal the twisted spin sum
    \begin{equation}
        \Pi^{AB,CD}  (\mathbf{p},\sqrt{\mathbf{p}^2+m^2})
        =P_{}(\mathbf{p},\sqrt{\mathbf{p}^2+m^2})
    \end{equation}
    and it is either an odd function or an even function determined by the parity of $2A+2C$:
    \begin{equation}
        P(-\mathbf{p}, -p^0) = (-1)^{2A+2C} P(\mathbf{p},p^0)
    \end{equation}
    for all 4-vectors $p$. And the degrees $d$ of monomials in $P$ satisfies
    \begin{equation}
        \max\{|A-C|,|B-D|\} \le d/2 \le \min\{A+C,B+D\}
    \end{equation}
\end{thm}

For the proof we just modify the representation on $V$ by $M\mapsto D^{AB}(\Lambda)M D^{CD-1}(\Lambda)$, and the LHS of the transformation rule for $T$ matrices by $D^{AB}(\Lambda)T^{\cdots} D^{CD-1}$. 

Note that $C,D$ are interchanged, because $\Lambda\mapsto D^{CDT}(\Lambda^{-1})$ is isomorphic to $D^{CD}$, not $D^{DC}$ as in the spin sum case. 

At this point we can derive from \eqref{Field equation momentum} two field equations:
\begin{equation}\label{Field equation}
\begin{aligned}
    \psi^{(+)AB}(x) &= \Pi^{AB,CD}(-i\partial) \psi^{(+)CD}(x)\\
    \psi^{(-)AB}(x) &= (-)^{2A+2C}(-)^{2B+2D} \Pi^{AB,CD}(-i\partial) \psi^{(-)CD}(x)
\end{aligned}
\end{equation}
where we have abused the notation to let $\Pi^{AB,CD}(-i\partial)$ instead of $P(-i\partial)$ denote replacing $p_\mu$ by $-i\partial_\mu$ in the polynomial $P$. The two phases in the second line cancel, so we have the following one field equation:

\begin{thm}
(Field equation of any spin, massive case)
\footnote{\cite{tung1967relativistic} neglect the phase of negative energy part. }
\footnote{In the combination of positive and negative part, the phase $(-)^{2B}$ plays an essential  role. }

The following field equation is satisfied:
\begin{equation}\label{Field equation}
    \psi^{AB}(x) = \Pi^{AB,CD}(-i\partial) \psi^{CD}(x)
\end{equation}
\end{thm}

So clearly the field equation is not unique: we can adjust the polynomial $P$ by replacing $(p_0)^2$ with $\mathbf{p}^2+m^2$ in any monomial, which do not change the value on the mass shell. But this is just adding a term of Klein-Gordan operator connected by several partial derivative operators. 

The degrees of monomials in $P$ satisfy:
\begin{equation}
        \max\{|A-C|,|B-D|\} \le d/2 \le \min\{A+C,B+D\}
\end{equation}
As a special case, the degrees of $\Pi^{j,0;j,0}$ and $\Pi^{0,j;0,j}$ are all 0. So we do not have a nontrivial field equation for $(j,0)$ or $(0,j)$ field with itself from the above procedure. 
\footnote{This coincides with Weinberg.}
And $\Pi^{j,0;0,j}$ and $\Pi^{0,j;j,0}$ must be homogeneous of degree $2j$. This is the case that we have a homogeneous (except one term) field equation. 
\footnote{This is the result in Weinberg.}

But in the general case, the polynomial may not be homogeneous. 
\\\\\\

Let's talk about the relation between ordinary spin sum and twisted spin sum. As is used in the proof of the field equation, $D^{CD-T}$ is isomorphic (but not generally equal) $D^{CD}$, and $D^{CD*}$ is isomorphic to $D^{DC}$. So we may expect that $\pi^{AB,CD}(p)$ is almost the same as $\Pi^{AB,DC}(p)$, because their Lorentz transformation rule is under the same representation. 

Define 
\begin{equation}
    \Omega_{CD}: 
    \begin{pmatrix}
    \mathbb{C}^{(2C+1)(2D+1)} \to \mathbb{C}^{(2C+1)(2D+1)}\\
    e_{(c-1)(2D+1)+d} \mapsto e_{(d-1)(2C+1)+c}
    \end{pmatrix}
\end{equation}
This is just a basis transformation
\begin{equation}
    \Omega_{CD}: 
    \begin{pmatrix}
    V^{C}\otimes V^{D} \to V^{D}\otimes V^{C}\\
    v\otimes w \mapsto w\otimes v
    \end{pmatrix}
\end{equation}

So clearly 
\begin{equation}
\Omega_{CD}u^{CD}(\mathbf{0})=u^{DC}(\mathbf{0})
\end{equation}
Also we can show that $\Omega_{CD}$ is a homomorphism of representations:
\begin{equation}
    D^{DC-\dagger}(\Lambda) \Omega_{CD} = \Omega_{CD} D^{CD}(\Lambda)
\end{equation}
This is because $\exp(-ig)^{-\dagger} = \exp(-ig^\dagger)$ implies that the Lie algebra representation of $D^{DC-\dagger}$ is:
\begin{equation}
    \begin{aligned}
        \mathbf{J} \mapsto& \mathbf{J}^D \otimes I^C + I^D \otimes \mathbf{J}^C\\
        \mathbf{K} \mapsto& +i(\mathbf{J}^D \otimes I^C - I^D \otimes \mathbf{J}^C)
    \end{aligned}
\end{equation}
So putting these together we have
\begin{equation}
\begin{aligned}
    u^{DC\dagger}(\mathbf{0}) D^{DC-1}(Lp) &= u^{CD\dagger}(\mathbf{0})
    \Omega_{CD}^{\dagger}\Omega_{CD}^{-\dagger} D^{CD\dagger}(Lp) \Omega_{CD}^\dagger\\
    &= u^{CD\dagger}(\mathbf{0}) D^{CD\dagger}(Lp) \Omega_{CD}^\dagger
\end{aligned}
\end{equation}
Comparing with what we called the spin sum
\begin{equation}
    \pi^{AB,CD}(\mathbf{p})
    \equiv (2 p^0) u^{AB}(\mathbf{p}) u^{CD\dagger}(\mathbf{p})
    = (2 m) D^{AB}(Lp) u^{AB}(\mathbf{0}) u^{CD\dagger}(\mathbf{0}) D^{CD\dagger}(Lp)
\end{equation}
we derive:
\begin{prop}(Relation between spin sum and twisted spin sum)

\begin{equation}
    \Pi^{AB,DC}(p) = \frac{1}{2m} \pi^{AB,CD}(p) \Omega_{CD}^\dagger
\end{equation}
\end{prop}

Now we can formulate the field equation in terms of ordinary spin sum:

\begin{thm}(Field equation of any spin, massive case)

The following field equation is satisfied:
\begin{equation}\label{Field equation}
    \psi^{AB}(x) = \frac{1}{2m} \pi^{AB,DC}(-i\partial) \Omega_{DC}^\dagger \psi^{CD}(x)
\end{equation}
\end{thm}

In the special case where one of $C,D$ is 0, $\Omega_{CD}$ is the identity matrix, so the twisted spin sum is the same as the spin sum. This is the case in \cite{weinberg1964feynman}. This can be seen more directly by noticing
\begin{equation}
    \begin{aligned}
        D^{0,j}(\Lambda) = D^{j,0. -\dagger}(\Lambda)\\
        D^{j,0}(\Lambda) = D^{0,j. -\dagger}(\Lambda)
    \end{aligned}
\end{equation}
is not only isomorphic as representations, but directly equal in their standard basis.

\subsection{Examples}

The first example is massive spin-$\frac{1}{2}$ particles with field rep $(\frac{1}{2},0)$ or $(0, \frac{1}{2})$. The corresponding field equation is what we called Weyl equations or Dirac equation.

We can easily see that the zero momentum value coefficients and spin sums are:
\begin{equation}
    \begin{pmatrix}
        u^{1/2,0}(\mathbf{0})\\
        u^{0,1/2}(\mathbf{0})
    \end{pmatrix}
    =
    \begin{pmatrix}
        1 & 0\\
        0 & 1\\
        1 & 0\\
        0 & 1
    \end{pmatrix}
\end{equation}
\begin{equation}
    \begin{pmatrix}
        \pi^{1/2,0,1/2,0}(\mathbf{0}) & \pi^{1/2,0,0,1/2}(\mathbf{0})\\
        \pi^{0,1/2,1/2,0}(\mathbf{0}) & \pi^{0,1/2,0,1/2}(\mathbf{0})
    \end{pmatrix}
    = 2m
    \begin{pmatrix}
        I_2 & I_2 \\
        I_2 & I_2
    \end{pmatrix}
\end{equation}
By \eqref{ABCDK}, the the upper-left term and the lower-right term when the momentum varies equal a polynomial of degree $1$, while lower-left term and the upper-right term when the momentum varies equal polynomials of degree 0. And by the $T$ matrices for $(1/2,0,1/2,0)$ and $(0,1/2,0,1/2)$ computed in section 2, the expansion of zero momentum spin sums are:
\begin{equation}
    \pi^{1/2,0,1/2,0}(\mathbf{0}) = 2m\sigma^0, \quad \pi^{0,1/2,0,1/2}(\mathbf{0}) = 2m\bar{\sigma}^0
\end{equation}
so we have
\begin{equation}
    \begin{pmatrix}
        \pi^{1/2,0,1/2,0}(p) & \pi^{1/2,0,0,1/2}(p)\\
        \pi^{0,1/2,1/2,0}(p) & \pi^{0,1/2,0,1/2}(p)
    \end{pmatrix}
    = 2
    \begin{pmatrix}
        -p_\mu \sigma^\mu & m I_2 \\
        m I_2 & -p_\mu \bar{\sigma}^\mu
    \end{pmatrix}
\end{equation}

We abbreviate $\varphi = \psi^{1/2,0},\chi = \psi^{0,1/2}$, and the field equations \eqref{Field equation} write:
\begin{equation}
    \begin{aligned}
        \varphi = \frac{1}{2m} \pi^{1/2,0,0,1/2}(-i\partial) \varphi \\
        \varphi = \frac{1}{2m} \pi^{1/2,0,1/2,0}(-i\partial) \chi \\
        \chi = \frac{1}{2m} \pi^{0,1/2,1/2,0}(-i\partial) \chi \\
        \chi = \frac{1}{2m} \pi^{0,1/2,0,1/2}(-i\partial) \varphi
    \end{aligned}
\end{equation}
The first and the third line are trivial:
\begin{equation}
    \begin{aligned}
        \varphi &=  \varphi \\
        \chi &=  \chi
    \end{aligned}
\end{equation}
and the other two lines are:
\begin{equation}
    \begin{aligned}
        m \varphi &= i \sigma^\mu \partial_\mu \chi \\
        m \chi &= i \bar{\sigma}^\mu \partial_\mu \varphi
    \end{aligned}
\end{equation}
They are exactly Weyl equations for massive particles.

The second example is massive spin-$1$ particles with field rep $(\frac{1}{2},\frac{1}{2})$. 
The corresponding field equation is what we called Proca equation, also called the massive Maxwell equations. 

The field rep $(\frac{1}{2},\frac{1}{2})$ is not directly equal, but is equivalent to the vector representation (or called canonical representation) $D(\Lambda)=\Lambda$ of $SO^+(1,3)$. We take this form. 

The spin-$1$ rep of $su(2)$ is defined as:
\begin{equation}
    J_z \mapsto \begin{pmatrix}
        1 & 0 & 0\\
        0 & 0 & 0\\
        0 & 0 & -1
    \end{pmatrix}
    \quad\quad\quad
    J_+ \mapsto \begin{pmatrix}
        0 & \frac{1}{\sqrt{2}} & 0\\
        0 & 0 & \frac{1}{\sqrt{2}}\\
        0 & 0 & 0
    \end{pmatrix}
    \quad\quad\quad
    J_- \mapsto \begin{pmatrix}
        0 & 0 & 0\\
        \frac{1}{\sqrt{2}} & 0 & 0\\
        0 & \frac{1}{\sqrt{2}} & 0
    \end{pmatrix}
\end{equation}
and the vector representation maps everything to itself:
\begin{equation}
    J_z = -i \begin{pmatrix}
        0 & 0 & 0 & 0\\
        0 & 0 & 1 & 0\\
        0 & -1 & 0 & 0\\
        0 & 0 & 0 & 0
    \end{pmatrix}
    \quad\quad\quad
    J_\pm = \frac{-i}{\sqrt{2}} \begin{pmatrix}
        0 & 0 & 0 & 0\\
        0 & 0 & 0 & \mp i\\
        0 & 0 & 0 & 1\\
        0 & \pm i & -1 & 0
    \end{pmatrix}
\end{equation}
Under the vector representation, the eigenvectors of z-axis angular-momentum is:
\begin{equation}
    e_0 = \begin{pmatrix}
        0\\
        0\\
        0\\
        1
    \end{pmatrix}
    \quad\quad\quad
    e_+ = J_+ e_0 = -\frac{1}{\sqrt{2}} \begin{pmatrix}
        0\\
        1\\
        i\\
        0
    \end{pmatrix}
    \quad\quad\quad
    e_- = J_- e_0 = \frac{1}{\sqrt{2}} \begin{pmatrix}
        0\\
        1\\
        -i\\
        0
    \end{pmatrix}
\end{equation}
So the initial value of the coefficients is:
\begin{equation}
    u(0) = \begin{pmatrix}
        0 & 0 & 0\\
        -\frac{1}{\sqrt{2}} & 0 & \frac{1}{\sqrt{2}}\\
        -\frac{i}{\sqrt{2}} & 0 & -\frac{i}{\sqrt{2}}\\
        0 & 1 & 0
    \end{pmatrix}
\end{equation}
and the initial value of the spin sum is:
\begin{equation}
    \pi(0) = 2m \begin{pmatrix}
        0 & 0 & 0 & 0\\
        0 & 1 & 0 & 0\\
        0 & 0 & 1 & 0\\
        0 & 0 & 0 & 1
    \end{pmatrix}
\end{equation}
We can show that
\begin{equation}
    {}^{vec,vec}_{1}(T^{\mu\rho})^{\nu\sigma} = g^{\mu\sigma}g^{\nu\rho}
\end{equation}
Indeed we should verify:
\begin{equation}
    \Lambda^\mu{}_\nu (T^{ab})^\nu{}_\rho (\Lambda^{-1})^\rho{}_\sigma = \Lambda_\rho{}^a \Lambda_\nu{}^b (T^{\rho\nu})^\mu{}_\sigma
\end{equation}
This is equivalent to
\begin{equation}
    \Lambda^\mu{}_\nu g^{a\nu} g^b{}_\rho \Lambda_\sigma{}^\rho = \Lambda_\rho{}^a \Lambda_\nu{}^b g^{\mu\rho}g^\nu{}_\sigma
\end{equation}
which is obvious. Also clearly we have
\begin{equation}
    {}^{vec,vec}_0 T = I_4
\end{equation}
So the initial value of the spin sum is decomposed as:
\begin{equation}
    \pi(0) = 2m[I_4 - \begin{pmatrix}
        1 & 0 & 0 & 0\\
        0 & 0 & 0 & 0\\
        0 & 0 & 0 & 0\\
        0 & 0 & 0 & 0
    \end{pmatrix}]
    = 2m [ {}^{vec,vec}_0(T)^{\circ\circ} - {}^{vec,vec}_1 (T^{\mu\rho})^{\circ\circ} ]
\end{equation}
So the field equation is 
\begin{equation}
    (I_4 - \frac{\partial_\mu \partial_\rho}{m^2}{}^{vec,vec}_1(T^{\mu\rho})^{\circ\circ}) B_\circ = B_\circ
\end{equation}
It is equivalent to:
\begin{equation}
    \partial_\mu \partial_\rho g^{\mu\sigma}g^{\nu\rho} B_\sigma = 0
\end{equation}
So the field equation becomes:
\begin{equation}
    \partial^\nu \partial^\mu B_\mu = 0
\end{equation}
Taking another $\partial_\nu$ and using Klein-Gordan equation, we have:
\begin{equation}\label{Gauge}
    \partial^\mu B_\mu = 0
\end{equation}
This, combined with Klein-Gordan equation, is equivalent to the famous Proca equation:
\begin{equation}
    \partial_\mu(\partial^\mu B^\nu - \partial^\nu B^\mu) + m^2 B^\nu = 0
\end{equation}
\begin{rmk}
Notice that for vector field of massless particles with helicity-$1$, \eqref{Gauge} is a gauge fixing condition we artificially add to the field, while in the massive case it is a field equation that must be satisfied. So the vector field for massive spin-$1$ particles do not have gauge invariance. 
\end{rmk}


\section{Conclusions}

In this article we proved spin statistics for arbitrary massive $(A,B)$ field. The quantity 'spin sum' $\pi^{AB}(p)$ is also used to derive the massive field equations, after a little modification. Klein-Gordan equation is just a condition of 4-momentum in this context. 
Explicit calculation shows that (Massive) Weyl equation, Dirac equation and Proca equation are all examples of our general field equation. 

In the case of $(A,B,C,D) = (j,0,0,j)$ or $(A,B,C,D) = (0,j,j,0)$, the polynomial of the twisted spin sum is homogeneous, so the field equation contains a constant term plus a homogeneous part, thanks to the condition on $K$. But in the general case it may not be homogeneous. 

In this article we just proved the 'existence' part of the polynomial of the spin sum, what remains undiscovered is an explicit formula for the coefficients $\xi_{K}^{ABCD}$. Weinberg in \cite{weinberg1969feynman} discussed its explicit formula. 

Also in the massless case the proof in this article cannot be directly used. One solution is to derive and explicit expansion for the spin sum $\pi^{j,0}(p)$ and $\pi^{0,j}(p)$, as in \cite{weinberg1964feynman2}, and then pass to the $m\to 0$ limit. Another is to notice that $(A,B)$ fields can be constructed out of $(j,0)$ and $(0,j)$ fields. 

Also, Rarita-Schwinger equation remains undiscovered in this article.

\clearpage

\nocite{*}
\bibliographystyle{alpha}
\bibliography{ref.bib}

\begin{thebibliography}{Wei64b}

\bibitem[Fen23]{feng2023representation}
Zixuan Feng.
\newblock Representation theory in the construction of free quantum field.
\newblock {\em arXiv preprint arXiv:2302.13808}, 2023.

\bibitem[Tun67]{tung1967relativistic}
Wu-Ki Tung.
\newblock Relativistic wave equations and field theory for arbitrary spin.
\newblock {\em Physical Review}, 156(5):1385, 1967.

\bibitem[Wei64a]{weinberg1964feynman}
Steven Weinberg.
\newblock Feynman rules for any spin.
\newblock {\em Physical Review}, 133(5B):B1318, 1964.

\bibitem[Wei64b]{weinberg1964feynman2}
Steven Weinberg.
\newblock Feynman rules for any spin. ii. massless particles.
\newblock {\em Physical Review}, 134(4B):B882, 1964.

\bibitem[Wei69]{weinberg1969feynman}
Steven Weinberg.
\newblock Feynman rules for any spin. iii.
\newblock {\em Physical Review}, 181(5):1893, 1969.

\bibitem[Wei02]{weinberg2002quantum}
Steven Weinberg.
\newblock {\em The quantum theory of fields: Foundations}.
\newblock Cambridge University Press, 2002.

\end{thebibliography}

\end{document}